\documentclass[11pt,letterpaper]{article}
\usepackage[letterpaper,margin=1.05in]{geometry}
\usepackage[T1]{fontenc}
\usepackage[utf8]{inputenc}
\usepackage{mathpazo}
\usepackage[scaled=0.92]{helvet}
\usepackage{courier}
\usepackage{microtype}
\usepackage{amsmath,amssymb,amsthm,mathtools,bm}
\usepackage{booktabs,array,tabularx,longtable}
\usepackage{xcolor,graphicx,tikz}
\graphicspath{{large_library_source/figures/}{figures/}}
\usepackage{enumitem}
\usepackage{hyperref}
\usepackage[nameinlink,capitalise]{cleveref}
\usepackage{fancyhdr}
\usepackage{titlesec}
\usepackage{caption}
\usepackage{placeins}
\newif\ifarxiv \arxivtrue

\definecolor{ink}{HTML}{17213A}
\definecolor{accent}{HTML}{A36A00}
\definecolor{muted}{HTML}{5D6472}
\definecolor{rulecolor}{HTML}{D5D9E2}
\hypersetup{colorlinks=true,linkcolor=ink,citecolor=accent,urlcolor=accent,
  pdftitle={Large Signal Libraries: Equal-Weight Limits and the Divergent Spectra of Signals and PnL},
  pdfauthor={Marc Nunes},
  pdfsubject={Population limits for expanding libraries of cross-sectional trading signals}}
\titleformat{\section}{\Large\bfseries\color{ink}}{\thesection.}{0.55em}{}
\titleformat{\subsection}{\large\bfseries\color{ink}}{\thesubsection.}{0.5em}{}
\titlespacing*{\section}{0pt}{2.2ex plus .5ex}{1ex}
\titlespacing*{\subsection}{0pt}{1.6ex plus .4ex}{.7ex}
\setlist{nosep,leftmargin=1.5em}
\renewcommand{\headrulewidth}{0.35pt}
\renewcommand{\headrule}{\hbox to\headwidth{\color{rulecolor}\leaders\hrule height \headrulewidth\hfill}}

\newtheorem{theorem}{Theorem}[section]
\newtheorem{proposition}[theorem]{Proposition}
\newtheorem{corollary}[theorem]{Corollary}
\newtheorem{lemma}[theorem]{Lemma}
\theoremstyle{definition}
\newtheorem{definition}[theorem]{Definition}
\newtheorem{example}[theorem]{Example}
\theoremstyle{remark}

\newcommand{\E}{\mathbb E}
\newcommand{\R}{\mathbb R}
\newcommand{\cH}{\mathcal H}

\newcommand{\cT}{\mathcal T}
\newcommand{\ip}[2]{\left\langle #1,#2\right\rangle}
\newcommand{\norm}[1]{\left\lVert #1\right\rVert}
\newcommand{\one}{\mathbf 1}
\newcommand{\sig}{\mathrm{sig}}
\newcommand{\pnl}{\mathrm{pnl}}
\newcommand{\perpcomp}{\mathrm{\perp}}
\DeclareMathOperator{\Cov}{Cov}
\DeclareMathOperator{\Var}{Var}

\DeclareMathOperator{\spanop}{span}

\begin{document}
\hypersetup{pageanchor=false}

\begin{titlepage}
\color{ink}
\vspace*{0.65in}
{\sffamily\bfseries\fontsize{29}{34}\selectfont Large Signal Libraries\par}
\vspace{0.16in}
{\sffamily\fontsize{16}{21}\selectfont Equal-Weight Limits and the Divergent Spectra of Signals and PnL\par}
\vspace{0.32in}
{\color{accent}\rule{1.45in}{2pt}}\par
\vspace{0.34in}
{\large Marc Nunes}\par
\vspace{0.05in}
{\small AlphaNova \quad $\cdot$ \quad \href{mailto:marc.nunes@alphanova.tech}{marc.nunes@alphanova.tech}}\par
\vspace{0.25in}
{\small Research Note \quad $\cdot$ \quad Version 10 \quad $\cdot$ \quad September 2026}\par
\vfill
\begin{minipage}{0.94\textwidth}
\small
\textbf{Abstract.}
An ensemble of roughly 3,000 signals over 20 assets was reported to have approximately 90\% correlation with the leading component of the asset-space return structure. Does having about 158 signals per available linear dimension explain that alignment? The population answer depends on the research process's design distribution and its relation to returns: crowding alone imposes neither a nonzero mean nor agreement with a principal component. The population theory developed here distinguishes four objects: the equal-weight signal, signal principal components, equal-weight profit and loss (PnL), and PnL principal components. The return operator in the motivating observation is a separate object. Independent libraries converge to their design mean; exchangeable libraries can retain a random conditional mean. Cross-sectional signals over $d$ assets, once demeaned and normalized, lie on the unit sphere $S^{q-1}$ of a $q$-dimensional space, $q=d-1$. Under axial symmetry, their nonzero mean is signal-cloud PC1 exactly when its longitudinal second moment exceeds $1/q$, the isotropic energy share. Residual-and-gap bounds quantify approximate alignment. Combining design weights into signals contracts the tangent of their angle to PC1 to at most $\sqrt{\lambda_2/\lambda_1}$ times its value, where $\lambda_1>\lambda_2$ are the leading signal-kernel eigenvalues; the factor is sharp. A target-aligned frame separates transverse signal geometry, which PnL discards, from dispersion weighting and temporal centering, which also change the spectrum. A reproducible synthetic example illustrates the geometric threshold, and a proposed empirical program separates library growth from limited-history estimation. No market-data empirical results are presented.
\end{minipage}
\vfill
{\footnotesize\color{muted}Keywords: cross-sectional signals; equal weighting; principal components; kernel operators; information coefficient; PnL covariance; exchangeability; large signal libraries; equal-weight ensembles.}
\end{titlepage}
\hypersetup{pageanchor=true}
\setcounter{page}{1}

\section{The question and the answer}

This paper began with an observation: an equal-weight ensemble of roughly 3,000 signals across 20 assets closely aligned with the leading principal component of the asset-space return structure, at approximately 90\% correlation. In that market and at that horizon, the leading direction was identifiable as reversal. This reported observation motivates the theory; the paper contains no market-data empirical results. Its numerical example is explicitly synthetic.

After cross-sectional demeaning, only 19 linear dimensions remain: about 158 submitted signals per available dimension, with every date's $3{,}000\times3{,}000$ signal Gram matrix having rank at most 19. Does this crowding explain mean--PC1 alignment, or does the alignment reveal how the research process populated the sphere? This is a question about the relation between a library mean and a dominant eigendirection; reversal names the observed instance.

We use EWS for the equally weighted sum, taken in its averaged form $K^{-1}\sum_i S_i$ unless stated otherwise; IC for the information coefficient; and PnL for profit and loss. Here IC is realized cross-sectional correlation; PC1 denotes a leading principal component.

Suppose a research process produces signals $S_1,S_2,\ldots$, each a normalized cross-sectional portfolio over the same assets. The simplest ensemble is
\[
  \bar S_K=\frac1K\sum_{i=1}^K S_i.
\]
The limit requires a rule for adding designs. If every entry copies $S_0$, the average is $S_0$ for every $K$, and its signal-cloud second moment is already rank one. Mean--PC1 agreement is then trivial, with no diversification or new direction (\cref{ex:duplicates}).

The interesting question is when a library of varying signals has a stable average and when that average aligns with a leading principal component. We must also say which population supplies the component: signals sampled across designs, or returns observed across dates. Equal weighting measures a \emph{first moment}, the mean; PCA measures a \emph{second moment}, the energy carried by different directions. Increasing $K$ can improve estimates of the signal population, but it neither forces these objects to agree nor improves return-covariance estimation, which requires more history.

Under a suitable sampling law, the basic limit is
\[
  \bar S_K \longrightarrow \mu:=\E_{Z}[S_Z]
\]
where $Z$ is drawn from the specified population of signal designs. Independent sampling from a fixed distribution is one sufficient assumption; \cref{sec:growth-assumptions,sec:diversity} explain the role of dependence and library composition. Signal PC1 is stabilized instead by convergence of second moments and a positive leading eigengap. Neither condition forces the mean into that eigendirection. Agreement with signal-cloud PC1 also does not establish agreement with return PC1; \cref{sec:asset-pca,sec:return-pca} distinguish the operators.

PnL adds a second separation. At date $t$, let $Q_t$ be the normalized cross-sectional return direction and $a_t$ its magnitude. A normalized signal has realized IC $p_{Z,t}=\ip{S_{Z,t}}{Q_t}$ and PnL
\[
  \Pi_{Z,t}=a_t p_{Z,t}.
\]
PnL depends only on the signal's projection onto the return direction, multiplied by return magnitude. The component orthogonal to $Q_t$ contributes nothing to that date's PnL. Signal PCA and PnL PCA therefore describe different operators.

The results below identify the library limit, give exact and approximate conditions for mean--PC alignment, and track the passage from signal geometry to PnL. The roadmap in \cref{tab:taxonomy,tab:assumptions} collects the objects and assumptions. These are statements about population geometry, not a claim that equal weighting is optimal.

\subsection{Relation to existing theory}

The identity $\Pi=a\,\mathrm{IC}$ connects this question to the fundamental-law tradition in active management \cite{grinold1994,qianhua2004,dingmartin2017}. The moving-frame decomposition was developed for finite signal libraries in \cite{nunes2026}. Here the question is how averages and principal components behave as the library grows.

Equal-weight portfolios and forecast combinations motivate a related but different question: how estimation error affects the performance of simple and optimized weights \cite{demiguel2009,batesgranger1969,timmermann2006}. Here we ask what an expanding library averages and when that mean aligns with a principal component; our limit results do not rank trading or forecasting performance. Corrections for multiple testing in factor discovery and for data snooping matter because the research measure is selected using data \cite{harvey2016,white2000}. Reversal has its own empirical literature \cite{jegadeesh1990,lehmann1990}; those findings contextualize the observed factor label but do not establish the reported ensemble alignment.

Functional PCA provides the framework for treating whole signal histories as random elements of a Hilbert space \cite{dauxois1982,bosq2000}. Integral-operator approximation describes the limit of the library Gram matrix \cite{koltchinskii2000,rosasco2010}; Davis--Kahan perturbation theory controls changes in eigenvectors when an eigengap is present \cite{daviskahan1970}. Directional statistics provides the background for distributions on spheres \cite{mardia2000}.

Relative to \cite{nunes2026}, this paper moves from finite-library correlation identities to a specified law for adding signal designs. It identifies the resulting mean, including a random conditional mean for exchangeable libraries, compares that mean with signal and PnL principal components, and develops row-degree diagnostics whose bounds carry through to the combined signal or PnL history. The laws of large numbers, Gram/operator duality, and residual bounds are classical ingredients. The contribution here is their formulation for a growing research library, and the conditions and diagnostics that result. The sharp synthesis inequality in \cref{cor:synthesis-transport} makes explicit how a bound on design weights controls the ensemble itself.

\section{Finite cross-sections and an expanding library}\label{sec:growth-assumptions}

Let $d\ge 3$ be the number of assets. Demeaning subtracts each signal's cross-sectional average, imposing one linear constraint and leaving the zero-sum, or neutral, space
\[
 H=\{x\in\R^d:\one^\top x=0\},\qquad q:=\dim H=d-1.
\]
Normalizing a nonzero demeaned signal fixes its length at one. Its direction therefore lies on
\[
  S(H):=\{x\in H:\norm{x}=1\}\cong S^{q-1}=S^{d-2}.
\]
In the running example there are 20 assets. Demeaning leaves 19 linear dimensions ($q=19$); fixing the norm places the signals on the 18-dimensional sphere $S^{18}$ within that space. This modest dimension makes the longitudinal-energy threshold $1/q$ substantive, as the synthetic illustration in \cref{sec:date-alignment} shows. The theory keeps $q$ general. Compactness and unit norms bound the signals but do not specify how a research process distributes mass on the sphere.

Let $(\cT,\tau)$ represent dates or market states, where $\tau$ is a probability measure---either a population law or the uniform measure on a fixed observed history. We treat each signal's entire history as one element of the Hilbert space
\[
  \cH_S:=L^2(\cT,\tau;H),\qquad
  \ip{f}{g}_{\cH_S}:=\int_{\cT}\ip{f_t}{g_t}_H\,d\tau(t).
\]
The inner product averages the two signals' cross-sectional similarity over dates. Throughout, $f\otimes g$ denotes the rank-one operator $h\mapsto\ip{g}{h}f$.

Let $\mathcal Z$ be the set of signal designs and $\nu$ their sampling distribution. This distribution records how often different designs occur under the research and selection process. A design $z\in\mathcal Z$ produces a measurable history $S_z\in\cH_S$, with $\norm{S_{z,t}}=1$ almost everywhere. The growing library $Z_1,Z_2,\ldots$ samples these designs. Our baseline is independent and identically distributed (iid) sampling from a fixed $\nu$. Ergodic or weakly dependent sequences may replace it when their averages satisfy the same law of large numbers.

Three requirements remain separate: population moments must stabilize; dependence must permit concentration around the claimed limit; and a principal-direction claim needs suitable second-moment structure and an eigengap. Iid sampling from a fixed distribution supplies the first two, even if that distribution consists of a single design. It does not supply the third. A changing mix of research families can instead produce a changing mean.

Exchangeability is weaker than independence: the joint law is unchanged when the signal indices are permuted. It allows the library to share a random latent environment. As \cref{prop:definetti} shows, the resulting mean limit can be random even when every design has the same marginal law.

We take $\nu$ to be a probability measure and work with separable $L^2$ spaces and jointly measurable versions of $(z,t)\mapsto S_{z,t}$. These assumptions make the Bochner integrals and the conditional sampling statements below well-defined.

\begin{definition}[Library mean and equal-weight ensemble]
The population mean signal and the $K$-signal equal-weight ensemble are
\[
  \mu:=\int_{\mathcal Z}S_z\,d\nu(z),
  \qquad
  \mu_K:=\frac1K\sum_{i=1}^K S_{Z_i}.
\]
Both are elements of $\cH_S$. Pointwise, $\mu_K(t)$ is the unnormalized equally weighted cross-sectional signal at date $t$.
\end{definition}

The distinction between the average $\mu_K(t)$ and its renormalized direction is important. If $\mu(t)\ne0$, then continuity of $x\mapsto x/\norm{x}$ yields convergence of directions. If $\mu(t)=0$, normalization is unstable and there is no population direction to estimate.

\begin{theorem}[Equal-weight limit]\label{thm:lln}
If $S_{Z_i}$ are iid as $\cH_S$-valued random elements, then
\[
  \mu_K\longrightarrow\mu
  \quad\text{almost surely in }\cH_S.
\]
If dates are discrete or chosen pointwise versions are fixed, then for any $t$ at which an ordinary vector law of large numbers applies and $\mu(t)\ne0$,
\[
  \frac{\sum_{i=1}^K S_{Z_i,t}}{\norm{\sum_{i=1}^K S_{Z_i,t}}}
  \longrightarrow
  \frac{\mu(t)}{\norm{\mu(t)}}.
\]
\end{theorem}

\begin{proof}
The Hilbert-space strong law requires a finite first moment. Here $\E\norm{S_Z}_{\cH_S}=1$, so it applies directly and gives the first claim. The second is the continuous-mapping theorem applied pointwise.
\end{proof}

\begin{proposition}[Exchangeable-library limit]\label{prop:definetti}
Suppose $(S_{Z_i})_{i\ge1}$ is an infinitely exchangeable sequence of $\cH_S$-valued random elements and let $\mathcal G$ be its exchangeable sigma-field. Unit history norms supply the required integrability automatically, and
\[
  \mu_K\longrightarrow \E[S_{Z_1}\mid\mathcal G]
  \quad\text{almost surely in }\cH_S.
\]
The limit is generally random. It equals the unconditional mean $\mu$ only when the directing measure has an almost-surely constant mean; iid sampling is the canonical special case.
\end{proposition}

\begin{proof}
Conditional on the de Finetti directing measure, the sequence is iid. The conditional Hilbert-space strong law gives convergence to the directing measure's mean, which is $\E[S_{Z_1}\mid\mathcal G]$; see \cite{kallenberg2005} for the general exchangeability framework.
\end{proof}

In practical terms, an exchangeable research process may first select a shared research environment and then produce designs within it. The infinite EWS records the mean within that environment, which need not be the unconditional mean across all possible environments.

These results identify the limit. Whether it is nonzero, economically useful, or a leading principal direction of signals or returns remains a separate question.

\subsection{Objects and assumptions at a glance}

\Cref{tab:taxonomy} is the four-object roadmap. In the notation below, $Z\sim\nu$ averages designs, whereas $t\sim\tau$ averages dates. The operator names record both the observation space and the averaging population:
\begin{itemize}
  \item $M_t$ and $C_t$: one date's signal second moment and covariance on the asset space $H$; $M_{\mathrm{ret}}$ and $C_{\mathrm{ret}}$: return second moment and covariance averaged over dates.
  \item $A_{\sig}$ and $A_{\pnl}$: second moments of whole signal histories and centered payoff histories. Their Gram counterparts $T_{\sig}$ and $T_{\pnl}$ act on design weights and have the same nonzero spectra.
  \item $\Phi_S$ and $\Phi_\Pi$: maps that combine design weights into signal and centered payoff histories. $L$ projects signals onto returns and scales by return magnitude; $\mathcal L=P_0L$ also centers through time, producing features $\Psi_z=P_0\Pi_z$.
\end{itemize}
These maps are defined in \cref{sec:asset-pca,sec:return-pca,sec:pnl-map}. \Cref{tab:assumptions} separates assumptions that give a limit from those that make it a principal direction. In particular, a library law of large numbers supplies no link to return PC1.

\begin{table}[!htbp]
\centering
\caption{The four large-library objects and what determines them.}
\label{tab:taxonomy}
\small
\begin{tabularx}{\textwidth}{@{}>{\bfseries\raggedright\arraybackslash}p{0.22\textwidth}X X@{}}
\toprule
Object & Population limit & Information retained \\
\midrule
Equal-weight signal & $\mu(t)=\E_ZS_{Z,t}$ & First moment; longitudinal and mean transverse components \\
Signal PCA & $T_{\sig}$ or $A_{\sig}$ & Second-order longitudinal plus transverse similarity \\
Equal-weight PnL & $a_t\E_Zp_{Z,t}$ & First moment of realized IC, scaled by dispersion \\
PnL PCA & $T_{\pnl}$ or $A_{\pnl}$ & Centered, dispersion-weighted longitudinal dependence only \\
\bottomrule
\end{tabularx}
\end{table}

\begin{table}[!htbp]
\centering
\caption{Assumption roadmap. Every PCA claim refers to the specified operator.}
\label{tab:assumptions}
\small
\begin{tabularx}{\textwidth}{@{}>{\raggedright\arraybackslash}p{.31\textwidth}X@{}}
\toprule
Claim & What must additionally hold \\
\midrule
EWS has a stable limit & A library law of large numbers; a deterministic mean additionally excludes variation in the directing mean (\cref{thm:lln,prop:definetti}). \\
Normalized EWS converges & A pointwise law of large numbers and a nonzero limiting mean. \\
Sample PC1 converges & Operator convergence and a simple leading eigenvalue (\cref{thm:opconv}). \\
Mean is signal-cloud PC1 & Zero transverse residual and strict dominance (\cref{thm:meanpc}); under axial symmetry, longitudinal energy above $1/q$ (\cref{cor:axial}). \\
EWS approximately aligns with signal PC1 & Row-degree residual small relative to its gap, followed by synthesis into the relevant signal space (\cref{lem:approxcentrality,cor:synthesis-transport}). \\
Mean aligns with the target axis & Mean transverse component vanishes and mean IC is nonzero; coordinates alone do not ensure either. \\
Signal PC1 maps to PnL PC1 & Bounded payoff map, compatible eigenbasis, and product-eigenvalue dominance (\cref{prop:payoff-transport}). \\
Mean aligns with return PC1 & A further relation between the design mean and the chosen return operator; this is not a library-size consequence. \\
\bottomrule
\end{tabularx}
\end{table}

\FloatBarrier
\section{Why signal count is not effective diversity}\label{sec:diversity}

Signals produced by the same pipeline often share much of their variation. Counting them as separate submissions does not make that variation disappear when they are averaged.

At a fixed date, write $X_i=S_{Z_i,t}-\E S_{Z_i,t}$ and let
\[
  \Gamma_{ij}:=\E\ip{X_i}{X_j}.
\]
Then the exact mean-square fluctuation of the library average is
\begin{equation}\label{eq:mse}
  \E\norm{\mu_K(t)-\E\mu_K(t)}^2
  =\frac1{K^2}\sum_{i,j=1}^K\Gamma_{ij}.
\end{equation}
Equation \eqref{eq:mse} makes the averaging requirement precise: fluctuations vanish in mean square exactly when the summed covariances grow more slowly than $K^2$. Convergence to a particular deterministic mean $\mu(t)$ also requires $\E\mu_K(t)\to\mu(t)$. Thus stable composition and enough cancellation are separate assumptions; neither follows from the count alone.

If $\Gamma_{ii}=\sigma^2$ and $\Gamma_{ij}=\rho\sigma^2$ for $i\ne j$, then
\begin{equation}\label{eq:equicorr}
  \E\norm{\mu_K(t)-\mu(t)}^2
  =\sigma^2\left(\rho+\frac{1-\rho}{K}\right),
\end{equation}
where the second display assumes identical marginal laws, so $\E\mu_K(t)=\mu(t)$. For $\rho>0$, the averaging error has a nonzero floor. Equicorrelation is consistent with exchangeability but does not imply it. In the exchangeable case, letting $B_t=\E[S_{Z_1,t}\mid\mathcal G]$, conditional independence gives, for $i\ne j$,
\[
  \Gamma_{ij}=\E\norm{B_t-\mu(t)}^2.
\]
The error floor therefore comes from variation in the conditional \emph{mean}. The directing distribution may itself be random without producing a floor if its mean is always the same. In the equicorrelation model, a convenient effective count is
\[
  K_{\mathrm{eff}}
  :=\frac{K}{1+(K-1)\rho},
\]
which tends to $1/\rho$ when $\rho>0$. Beyond that point, more submissions bring little reduction in mean-square error.

\begin{example}[Exact duplication]\label{ex:duplicates}
If $S_1=S_2=\cdots=S_0$, then $\mu_K=S_0$ and $A_{\sig,K}=S_0\otimes S_0$ for every $K$. Both the average and its rank-one second moment stay unchanged. EWS and signal PC1 agree from the first observation; no diversification or new direction appears as the library grows. This does not contradict a law of large numbers. For a fixed $S_0$, the sampling distribution is simply concentrated on one design. If all entries copy the same random $S_0$, their average instead retains that randomness and need not converge to $\E S_0$.
\end{example}

\begin{example}[Families with unequal multiplicity]
Suppose there are $J$ distinct research families with family means $m_1,\ldots,m_J$, and $n_j(K)$ variants from family $j$. If the within-family averages converge to $m_j$ for every family with positive limiting weight and $n_j(K)/K\to w_j$, then
\[
  \mu_K\longrightarrow\sum_{j=1}^J w_jm_j.
\]
The weights $w_j$ record the relative submission frequencies of the families. Increasing one family's share changes the limit even if the added submissions are copies.
\end{example}

For empirical work, report growth by research family, cluster, or generator as well as by raw signal count. Such groups are candidates for meaningful sampling units, but their labels alone do not establish independence. The dependence that remains between groups must also be assessed.

\section{Signal PCA and the separate return operator}

The phrase ``PCA of the signals'' can refer to different operators. Distinguishing them prevents dimensional and interpretive errors.

\subsection{Asset-space PCA at one date}\label{sec:asset-pca}

At a fixed date define the uncentered population second moment
\[
  M_t:=\int_{\mathcal Z}S_{z,t}\otimes S_{z,t}\,d\nu(z):H\to H,
\]
and its empirical version $M_{K,t}=K^{-1}\sum_iS_{Z_i,t}\otimes S_{Z_i,t}$. Since $H$ has dimension $q$, this is an ordinary $q\times q$ matrix. The centered covariance is
\[
  C_t=M_t-\mu_t\otimes\mu_t,
  \qquad \mu_t:=\int S_{z,t}\,d\nu(z).
\]
Uncentered PCA detects directions carrying signal energy relative to the origin; centered PCA detects dispersion around the library mean. We study uncentered alignment first, but centering changes dominance rather than the eigenvector compatibility condition. For $m_t=\mu_t/\norm{\mu_t}$ and $P_t=I-m_tm_t^\top$,
\[
 P_tC_tm_t=P_tM_tm_t,\qquad P_tC_tP_t=P_tM_tP_t.
\]
Thus the same residual determines whether $m_t$ is an eigenvector of either operator. When it is, centering subtracts $\norm{\mu_t}^2$ from its eigenvalue and leaves the transverse block unchanged. The strict dominance condition below must therefore subtract $\norm{\mu_t}^2$ on its left side when applied to $C_t$. A mean direction can be uncentered PC1 and fail to lead after centering.

\begin{theorem}[Exact mean--PC condition]\label{thm:meanpc}
Fix a date with $\mu_t\ne0$ and put $m_t=\mu_t/\norm{\mu_t}$. Then $m_t$ is the unique leading eigenvector of $M_t$, up to sign, if and only if
\begin{align}
  &(I-m_tm_t^\top)M_tm_t=0,\label{eq:eigcondition}\\
  &m_t^\top M_tm_t>
    \lambda_{\max}\!\left((I-m_tm_t^\top)M_t(I-m_tm_t^\top)\right).
    \label{eq:gapcondition}
\end{align}
\end{theorem}

\begin{proof}
Equation \eqref{eq:eigcondition} is exactly the condition that $M_tm_t$ have no component orthogonal to $m_t$. Under it, $H=\spanop(m_t)\oplus m_t^\perp$ reduces $M_t$. Inequality \eqref{eq:gapcondition} says the eigenvalue on $m_t$ strictly exceeds every eigenvalue on the orthogonal block.
\end{proof}

Writing $S_{z,t}=a_zm_t+U_z$ with $U_z\perp m_t$, condition \eqref{eq:eigcondition} becomes $\E[a_zU_z]=0$, and \eqref{eq:gapcondition} becomes
\[
  \E[a_z^2]>\lambda_{\max}(\E[U_z\otimes U_z]).
\]
The first condition removes coupling between the mean direction and transverse variation. The second requires the mean direction to carry more energy than any transverse direction. Axial symmetry is one way to remove the coupling, but the conditions do not require that symmetry.

\begin{corollary}[Axially symmetric population on $S^{q-1}$]\label{cor:axial}
Let $q=\dim H\ge2$ and $\mu_t\ne0$. Suppose the date-$t$ population is invariant under every orthogonal transformation of $H$ fixing $m_t=\mu_t/\norm{\mu_t}$. Then $m_t$ is the unique PC1 of $M_t$, up to sign, exactly when
\[
  \E[(m_t^\top S_{Z,t})^2]>\frac1q.
\]
\end{corollary}

\begin{proof}
The assumed invariance gives $\E[aU]=0$ and $\E[U\otimes U]=\lambda_\perp(I-m_tm_t^\top)$. Taking traces in $a^2+\norm{U}^2=1$ yields $(q-1)\lambda_\perp=1-\E[a^2]$. Longitudinal dominance is $\E[a^2]>\lambda_\perp$, which rearranges to the stated threshold. The orthogonal-group formulation includes the transverse reflection needed when $q=2$.
\end{proof}

Since $\E_Za_Z=\norm{\mu_t}$, the exact variance decomposition gives
\[
  \E_Z[a_Z^2]=\Var_Z(a_Z)+\norm{\mu_t}^2
  =\Var_Z(a_Z)+\E_{Z,Z'}\ip{S_{Z,t}}{S_{Z',t}},
\]
where $Z,Z'$ are independent population draws. Hence under axial symmetry the unique-PC1 condition is exactly
\begin{equation}\label{eq:axial-variance}
  \Var_Z(a_Z)+\E_{Z,Z'}\ip{S_{Z,t}}{S_{Z',t}}>\frac1q.
\end{equation}
The threshold $1/q$ is the energy share of any direction under the uniform distribution on $S^{q-1}$. Under axial symmetry, the mean direction wins exactly when it carries more than this isotropic share. Equality makes the second moment isotropic; a smaller longitudinal share makes the transverse block dominant, and can also describe an anisotropic cloud. Thus the criterion is excess longitudinal energy, not merely departure from isotropy.

Mean pairwise correlation above $1/q$ is sufficient, but longitudinal variance can also carry the sum over the threshold. As $q$ increases, more transverse directions share the remaining energy, so the threshold falls. It becomes easier to exceed if longitudinal concentration stays bounded away from zero. If that concentration falls at the same rate as $1/q$, dimension alone does not settle the comparison. Throughout this corollary, symmetry is about the \emph{mean-signal pole} $m_t$; the realized-return pole $Q_t$ is a different direction.

\subsection{Return PCA averages a different population}\label{sec:return-pca}
For demeaned realized returns $\widetilde R_t=a_tQ_t$ with $\E_\tau a_t^2<\infty$, define the return second moment and temporal covariance on the same asset space $H$:
\[
 M_{\mathrm{ret}}:=\E_\tau[a_t^2Q_t\otimes Q_t],\qquad
 C_{\mathrm{ret}}:=M_{\mathrm{ret}}-\bar R\otimes\bar R,
 \quad \bar R:=\E_\tau[a_tQ_t].
\]
Cross-sectional demeaning does not force the temporal mean $\bar R$ to vanish, so these operators coincide only when $\bar R=0$. Write $v^{\mathrm{ret}}_1$ for a unit leading eigenvector of the chosen return operator when its leading eigenvalue is simple. The motivating observation concerns this return-derived direction. The signal operator $M_t$ instead averages designs at one date. \Cref{thm:meanpc,cor:axial} characterize alignment with PC1 of $M_t$; neither is a theorem about $M_{\mathrm{ret}}$ or $C_{\mathrm{ret}}$.

The mean signal enters realized returns through the equal-weight PnL limit $a_t\ip{\mu_t}{Q_t}$ in \cref{prop:ewpnl}. This scalar overlap does not imply alignment with return PC1. An empirical comparison must specify temporal centering and the return-estimation window: a fixed window gives a fixed $v^{\mathrm{ret}}_1$, while rolling windows produce a sequence of directions.

\subsection{Process-space PCA}

Across all dates, define
\[
  A_{\sig}:=\int_{\mathcal Z}S_z\otimes S_z\,d\nu(z):\cH_S\to\cH_S.
\]
This positive trace-class operator treats each entire signal history as one observation for uncentered PCA. Its centered version is $A_{\sig}-\mu\otimes\mu$. Unlike $M_t$, it can capture how signal directions change together over time.

\subsection{Library-kernel PCA}

Define the signal-similarity kernel
\[
  k_{\sig}(z,z'):=\ip{S_z}{S_{z'}}_{\cH_S}
  =\int_{\cT}\ip{S_{z,t}}{S_{z',t}}\,d\tau(t),
\]
and the integral operator on $L^2(\nu)$,
\[
  (T_{\sig}f)(z):=\int_{\mathcal Z}k_{\sig}(z,z')f(z')\,d\nu(z').
\]
The kernel is the average similarity between two signal histories. The operator $T_{\sig}$ is the population counterpart of the scaled Gram matrix $K^{-1}[k_{\sig}(Z_i,Z_j)]$. Its eigenfunctions assign weights to designs, rather than to assets.

To turn a weight function into a combined signal, introduce the synthesis operator
\[
  \Phi_S:L^2(\nu)\to\cH_S,
  \qquad \Phi_Sf=\int_{\mathcal Z}f(z)S_z\,d\nu(z).
\]
Then
\begin{equation}\label{eq:duality}
  T_{\sig}=\Phi_S^*\Phi_S,
  \qquad A_{\sig}=\Phi_S\Phi_S^*.
\end{equation}
The two operators have the same nonzero eigenvalues, with eigenvectors transported by $\Phi_S$ or $\Phi_S^*$. This is the infinite-library form of the familiar $X^\top X$ versus $XX^\top$ duality used in kernel PCA \cite{scholkopf1998}.

\begin{proposition}[When equal weights are a library PC]\label{prop:rowsum}
The constant weight function $\one$ is an eigenfunction of $T_{\sig}$ if and only if
\begin{equation}\label{eq:degree}
  d_{\sig}(z):=\int_{\mathcal Z}k_{\sig}(z,z')\,d\nu(z')
\end{equation}
is constant for $\nu$-almost every $z$. If that constant eigenvalue strictly exceeds the upper spectral edge $\sup\sigma(T_{\sig}|_{\one^\perp})$ (taken as zero if $\one^\perp=\{0\}$), then equal weights are the unique leading library PC. In that case $\mu=\Phi_S\one$ is a leading eigenvector of $A_{\sig}$ whenever $\mu\ne0$.
\end{proposition}

\begin{proof}
By definition, $T_{\sig}\one=d_{\sig}$. The first statement follows. The dominance condition makes $\one$ the unique top eigenfunction. Finally, $A_{\sig}\mu=\Phi_ST_{\sig}\one=\lambda\Phi_S\one=\lambda\mu$.
\end{proof}

Condition \eqref{eq:degree} says that every design has the same average similarity to the population. We call this \emph{equal centrality}. The next result measures how closely equal weights align with PC1 when the row sums are nearly, but not exactly, constant. The size of the discrepancy matters relative to the spectral gap. The bound is the classical Rayleigh-quotient residual estimate; see, for example, Saad \cite[Section~3.2.2]{saad2011} for Hermitian matrices. We give the compact self-adjoint argument and its kernel-degree interpretation here.

\begin{lemma}[Residual bound and approximate equal centrality]\label{lem:approxcentrality}
Let $T$ be a compact self-adjoint operator on a separable Hilbert space of dimension at least two, let $u$ be a unit vector, and set
\[
  \theta=\ip{u}{Tu},\qquad r=Tu-\theta u,
  \qquad \varepsilon=\norm{r}.
\]
Assume $T$ has a simple largest eigenvalue $\lambda_1$ with unit eigenvector $v_1$. Eigenvalues are compared by value, not modulus. Write $\lambda_2(T)=\sup\sigma(T|_{v_1^\perp})$ for the upper edge of the remaining spectrum, including zero when appropriate. If $\theta>\lambda_2(T)$, then
\begin{equation}\label{eq:approxalign}
  \sin\angle(u,v_1)
  \le \frac{\varepsilon}{\theta-\lambda_2(T)}.
\end{equation}
For $T=T_{\sig}$ and $u=\one$ (recall that $\nu$ is a probability measure), $Tu=d_{\sig}$, $\theta=\int d_{\sig}\,d\nu=\norm{\Phi_S\one}_{\cH_S}^2=\norm{\mu}_{\cH_S}^2\ge0$, and $\varepsilon$ is the $L^2(\nu)$ standard deviation of the kernel degree. Thus
\[
  \sin\angle(\one,v_1)
  \le
  \frac{\theta\,\mathrm{CV}(d_{\sig})}{\theta-\lambda_2(T_{\sig})}.
\]
The mean row degree is therefore the squared norm of the mean signal history. The same statement holds for $T_{\pnl}$, with $\theta=\norm{\Phi_\Pi\one}_{\cH_\Pi}^2\ge0$. In either case the gap hypothesis makes $\theta>0$, so CV has a positive denominator.
\end{lemma}

\begin{proof}
Expand $u=\sum_j\alpha_jv_j$ in an orthonormal eigenbasis. The residual has squared norm $\sum_j\alpha_j^2(\lambda_j-\theta)^2$. For $j\ge2$, $|\lambda_j-\theta|\ge\theta-\lambda_2$, so $\sum_{j\ge2}\alpha_j^2\le\varepsilon^2/(\theta-\lambda_2)^2$. The left-hand side of \eqref{eq:approxalign} is $(\sum_{j\ge2}\alpha_j^2)^{1/2}$.
\end{proof}

Angles throughout are acute angles between nonzero one-dimensional spans, so eigenvector signs are immaterial. Perron--Frobenius positivity for positive operators, when its additional hypotheses hold \cite{schaefer1974}, can make the leading \emph{weights over signal designs} positive, but does not make them equal and does not imply a long-only portfolio across assets.

\begin{corollary}[Synthesis contracts the principal angle]\label{cor:synthesis-transport}
Let $\Phi$ be a compact synthesis operator and let $T=\Phi^*\Phi$ have $\lambda_1>\lambda_2\ge0$ and unit leading eigenvector $v_1$. For any unit $u$ with $\alpha_1=\ip{u}{v_1}\ne0$, set
\[
  e=\Phi u,\qquad h_1=\frac{\Phi v_1}{\sqrt{\lambda_1}}.
\]
Then $h_1$ is a unit leading eigenvector of $\Phi\Phi^*$, $e\ne0$, and
\begin{equation}\label{eq:transport-tangent}
  \tan\angle(e,h_1)
  \le \sqrt{\frac{\lambda_2}{\lambda_1}}\,
       \tan\angle(u,v_1).
\end{equation}
In particular, for $\Phi=\Phi_S$ and $u=\one$, the synthesized vector is the mean signal process $\mu$. The same result applies to $\Phi_\Pi$ and the centered equal-weight PnL process.
\end{corollary}

\begin{proof}
Write $u=\alpha_1v_1+w$ with $w\perp v_1$. Then
\[
  \ip{\Phi w}{\Phi v_1}=\ip{w}{Tv_1}=0,\qquad
  \norm{\Phi v_1}^2=\lambda_1,\qquad
  \norm{\Phi w}^2=\ip{w}{Tw}\le\lambda_2\norm{w}^2.
\]
Taking the ratio of orthogonal to longitudinal lengths gives \eqref{eq:transport-tangent}; the leading-eigenvector statement follows from Gram/operator duality.
\end{proof}

The factor in \eqref{eq:transport-tangent} is sharp: if $v_2$ is a unit eigenvector with eigenvalue $\lambda_2$, equality holds for $u=\alpha v_1+\beta v_2$ with $\alpha\beta\ne0$ and $\alpha^2+\beta^2=1$, since the residual lies entirely in the $\lambda_2$-eigenspace. No smaller universal contraction factor is possible.

Combining \cref{lem:approxcentrality,cor:synthesis-transport}, if $b=\varepsilon/(\theta-\lambda_2)<1$, then
\begin{equation}\label{eq:transport-residual}
  \tan\angle(\Phi u,h_1)
  \le\sqrt{\frac{\lambda_2}{\lambda_1}}\,
       \frac{b}{\sqrt{1-b^2}}.
\end{equation}
A spectral gap helps in two ways: it controls alignment of the design weights, and the synthesis map gives less weight to the remaining eigendirections. If $\lambda_2=0$, every combined signal with a nonzero leading coefficient is exactly aligned. For $\Phi_S$, the angle compares \emph{whole signal histories}. To make a claim at one date, apply the corollary to $\Phi_t f=\int f(z)S_{z,t}\,d\nu(z)$ and its own Gram operator. Close alignment over a history need not hold at every date in that history.

\section{Large-library spectral convergence}

We next ask whether the sample operators estimate their population counterparts consistently. Because $\norm{S_z}_{\cH_S}=1$, the rank-one operators $S_z\otimes S_z$ are uniformly bounded in Hilbert--Schmidt norm.

\begin{theorem}[Empirical operator and PC convergence]\label{thm:opconv}
For iid signal designs, define
\[
  A_{\sig,K}:=\frac1K\sum_{i=1}^K S_{Z_i}\otimes S_{Z_i}.
\]
Then $A_{\sig,K}\to A_{\sig}$ almost surely in Hilbert--Schmidt norm. If the leading population eigenvalue is simple, with gap
\[
  \delta:=\lambda_1(A_{\sig})-\lambda_2(A_{\sig})>0,
\]
then the empirical leading eigendirection converges to the population leading eigendirection. Quantitatively, for a compatible choice of signs,
\[
  \sin\angle(\hat v_{1,K},v_1)
  \le \frac{2\norm{A_{\sig,K}-A_{\sig}}_{\mathrm{op}}}{\delta}.
\]
\end{theorem}

\begin{proof}
Apply the Hilbert-space strong law to the Hilbert--Schmidt-valued variables $S_Z\otimes S_Z$. For completeness, the following argument gives the Davis--Kahan bound directly on the Hilbert space \cite{daviskahan1970}; \cite{yuwangsamworth2015} gives a related population-gap formulation for finite matrices. Put $E=A_{\sig,K}-A_{\sig}$ and $e=\norm{E}_{\mathrm{op}}$. If $e\ge\delta/2$, the bound is trivial. Otherwise the variational principle gives $\hat\lambda_1\ge\lambda_1-e$. Projecting the empirical eigenvector equation onto $v_1^\perp$ and inverting there yields
\[
 \sin\angle(\hat v_{1,K},v_1)
 \le \frac{e}{\delta-e}\le\frac{2e}{\delta}.
\]
The inverse exists because the transverse spectral edge is at most $\lambda_2$, leaving a gap of at least $\delta-e$.
\end{proof}

The eigengap identifies a particular leading direction. If $\lambda_1=\lambda_2$, individual sample PCs may rotate within a tied leading eigenspace even when that subspace is consistently estimated. An empirical claim that PC1 stabilizes should therefore report the gap and its sensitivity to resampling.

For kernel PCA, the same sampled designs produce $K^{-1}[k(Z_i,Z_j)]$. For the measurable, separable feature map used here, each ordered eigenvalue converges to its population counterpart (with zeros appended if needed), by \cref{thm:opconv} and Gram/operator duality. General integral-kernel spectral approximation is developed in \cite{koltchinskii2000,rosasco2010}. The Gram/operator duality in \eqref{eq:duality} permits computation on whichever side is smaller: the $K\times K$ library side or the fixed/functional feature side.

\section{The target-aligned moving frame}

Write realized demeaned returns as $\widetilde R_t=a_tQ_t$, where $a_t=\norm{\widetilde R_t}$ and $Q_t\in S(H)$. Fix a unit vector $e_1\in H$ to serve as north pole. It is a basis vector within the neutral space, not the standard basis vector of $\R^d$. Choose an orthogonal map $G_t:H\to H$ with $G_tQ_t=e_1$, so that the target points north at every date.

This leaves freedom to rotate or reflect the transverse coordinates while keeping the pole fixed. That freedom is the residual group $O(q-1)$ acting on $e_1^\perp$. We need only a measurable choice of $G_t$ across dates; a globally continuous choice over the sphere is neither required nor generally available.

Every signal becomes
\begin{equation}\label{eq:frame}
  P_{z,t}:=G_tS_{z,t}=p_{z,t}e_1+U_{z,t},
  \qquad U_{z,t}\perp e_1,
\end{equation}
with
\[
  p_{z,t}=\ip{S_{z,t}}{Q_t}=\mathrm{IC}_{z,t},
  \qquad \norm{U_{z,t}}^2=1-p_{z,t}^2.
\]
The gauge rotates the \emph{target} to the north pole at each date; it does not flip or rotate each signal independently. The same $G_t$ acts on every signal at that date.

The signal kernel decomposes exactly:
\begin{equation}\label{eq:kerneldecomp}
  k_{\sig}(z,z')
  =\int p_{z,t}p_{z',t}\,d\tau(t)
   +\int\ip{U_{z,t}}{U_{z',t}}\,d\tau(t)
  =:k_{\parallel}(z,z')+k_{\perpcomp}(z,z').
\end{equation}
Hence
\[
  T_{\sig}=T_{\parallel}+T_{\perpcomp}.
\]
The first term records shared alignment with returns. The second records similarity in the directions orthogonal to returns. Together they give the operator version of the finite-library Gram decomposition.

When realized $|p_{z,t}|$ is small, almost all of that date's signal energy lies in $U_{z,t}$. A small average IC through time does not imply this at each date: large positive and negative ICs can cancel in the average. This observation concerns the return pole $Q_t$, whereas \cref{cor:axial} concerns the mean signal direction.

At a fixed date, the equal-weight signal in the moving frame is
\[
  G_t\mu_K(t)=\bar p_K(t)e_1+\bar U_K(t),
\]
where $\bar p_K=K^{-1}\sum_i p_{Z_i,t}$ and $\bar U_K=K^{-1}\sum_iU_{Z_i,t}$. If the library distribution is invariant under the residual $O(q-1)$ action, then $\E_ZU_{Z,t}=0$ and the population mean lies in $\spanop(Q_t)$. This symmetry is an assumption about the signals. Changing coordinates does not create it.

\section{The PnL map and its spectrum}\label{sec:pnl-map}

The raw PnL process of design $z$ is
\[
  \Pi_z(t)=\ip{S_{z,t}}{\widetilde R_t}=a_tp_{z,t}.
\]
Let $P_0$ denote temporal centering, $P_0f=f-\int f\,d\tau$, and define the centered PnL feature map
\[
  \Psi_z:=P_0(ap_z)\in\cH_\Pi:=L_0^2(\cT,\tau).
\]
The return assumption $a\in L^2(\tau)$ in \cref{sec:return-pca} and the unit-signal bound $|p_{z,t}|\le1$ already give
\begin{equation}\label{eq:payoff-moment}
  \int_{\mathcal Z}\norm{ap_z}_{L^2(\tau)}^2\,d\nu(z)
  \le\norm{a}_{L^2(\tau)}^2<\infty.
\end{equation}
Thus integrability over both dates and designs is automatic in this setting; no additional payoff-moment assumption is needed. This is weaker than bounded return magnitude and suffices for the feature operators below.
The PnL covariance kernel and its integral operator are
\begin{align}
  k_{\pnl}(z,z')&:=\ip{\Psi_z}{\Psi_{z'}}_{\cH_\Pi}
  =\Cov_\tau(a_tp_{z,t},a_tp_{z',t}),\label{eq:pnlkernel}\\
  (T_{\pnl}f)(z)&:=\int k_{\pnl}(z,z')f(z')\,d\nu(z').
\end{align}
If each $\Psi_z$ with nonzero variance is divided by its time-series standard deviation, the same construction yields the PnL \emph{correlation} operator on those designs. Correlation PCA and covariance PCA need not agree.

Define the uncentered payoff map $L$ on the domain
\[
  \mathcal D(L):=\{f\in\cH_S:a_t\ip{f_t}{Q_t}\in L^2(\tau)\}
\]
by
\[
  (Lf)(t):=a_t\ip{f_t}{Q_t},
\]
and put $\mathcal L:=P_0L$ on $\mathcal D(L)$ for the centered payoff map. Then $\Psi_z=\mathcal LS_z$ for almost every design. Independently of whether these payoff maps extend to bounded operators on all of $\cH_S$, define
\[
  \Phi_\Pi f:=\int f(z)\Psi_z\,d\nu(z),
  \qquad T_{\pnl}=\Phi_\Pi^*\Phi_\Pi,
  \qquad A_{\pnl}:=\Phi_\Pi\Phi_\Pi^*.
\]
By \eqref{eq:payoff-moment}, $\Phi_\Pi$ is Hilbert--Schmidt and these two positive operators are trace class.

\paragraph{Bounded-operator factorization.}
For the following identities, assume additionally that $\mathcal L$ extends to a bounded map $\cH_S\to\cH_\Pi$. This holds, for example, if $a\in L^\infty(\tau)$, including any fixed finite history with finite dispersion. Boundedness permits passage through the Bochner integral and taking adjoints, giving
\begin{equation}\label{eq:pnl-factor}
  \Phi_\Pi=\mathcal L\Phi_S,
  \qquad
  T_{\pnl}=\Phi_S^*\mathcal L^*\mathcal L\Phi_S,
  \qquad
  A_{\pnl}=\mathcal L A_{\sig}\mathcal L^*.
\end{equation}

Under \eqref{eq:payoff-moment} alone, the feature-map definitions $T_{\pnl}=\Phi_\Pi^*\Phi_\Pi$ and $A_{\pnl}=\Phi_\Pi\Phi_\Pi^*$ remain valid. The adjoint factorizations above need the additional boundedness assumption. An unbounded payoff map would require a separate treatment of operator domains and closures, which we do not undertake here.

\begin{theorem}[Transverse annihilation]\label{thm:annihilation}
If $V\in\cH_S$ satisfies $\ip{V_t}{Q_t}=0$ almost everywhere, then $\mathcal LV=0$. Consequently, $T_{\pnl}$ depends on the library only through the dispersion-weighted IC processes $ap_z$ and is invariant to arbitrary changes in transverse components that preserve $p_z$.
\end{theorem}

\begin{proof}
Pointwise, $(LV)(t)=a_t\ip{V_t}{Q_t}=0$, so $V\in\mathcal D(L)$ and temporal centering preserves zero. The kernel statement follows directly from $k_{\pnl}(z,z')=\ip{P_0(ap_z)}{P_0(ap_{z'})}$ and does not need boundedness on all of $\cH_S$.
\end{proof}

Signal and PnL spectra can therefore disagree for two reasons:
\begin{enumerate}
  \item \emph{transverse annihilation}: signal PCA sees $T_\parallel+T_\perpcomp$, while the payoff map kills $T_\perpcomp$;
  \item \emph{longitudinal transformation}: even if $U_z\equiv0$, PnL covariance temporally centers $ap_z$ and quadratically reweights dates through $a_t^2$.
\end{enumerate}
Either operation can change the ordering of principal components even when all population moments are known exactly.

\begin{proposition}[Equal-weight PnL limit]\label{prop:ewpnl}
For the uncentered equal-weight PnL,
\[
  \bar\Pi_K(t):=\frac1K\sum_{i=1}^K\Pi_{Z_i}(t)=a_t\bar p_K(t).
\]
Under iid design sampling and \eqref{eq:payoff-moment}, almost surely in $L^2(\tau)$,
\[
  \bar\Pi_K\longrightarrow L\mu
  =a_t\E_Z[p_{Z,t}].
\]
After temporal centering, $K^{-1}\sum_i\Psi_{Z_i}\to\mathcal L\mu$.
\end{proposition}

\begin{proof}
Apply the Hilbert-space strong law directly to the payoff processes $ap_{Z_i}$. Fubini's theorem identifies their mean with $a_t\E_Zp_{Z,t}$. Moreover, Jensen's inequality and \eqref{eq:payoff-moment} show that this function is in $L^2(\tau)$, so $\mu\in\mathcal D(L)$ and the mean is $L\mu$. Finally $P_0$ is bounded. No continuity assumption on $L$ over all of $\cH_S$ is used.
\end{proof}

Equal-weight PnL is therefore controlled by the cross-library mean IC process: $(L\mu)(t)=a_t\ip{\mu_t}{Q_t}$. It measures the mean signal's overlap with the realized return direction, rather than with the leading eigenvector of $C_{\mathrm{ret}}$ or $M_{\mathrm{ret}}$. It is not, in general, PnL PC1 either.

\begin{proposition}[When equal weights are a PnL PC]
The constant function $\one$ is an eigenfunction of $T_{\pnl}$ if and only if
\[
  d_{\pnl}(z):=\int_{\mathcal Z}\Cov_\tau(a p_z,a p_{z'})\,d\nu(z')
\]
is constant almost everywhere. It is PC1 if that eigenvalue dominates $T_{\pnl}$ on $\one^\perp$.
\end{proposition}

This is a PnL equal-centrality condition, distinct from its signal-space counterpart. One can hold while the other fails.

\section{A commutative picture---and where it breaks}

Throughout this section assume that the centered payoff map $\mathcal L:\cH_S\to\cH_\Pi$ is bounded, as in the bounded-factorization paragraph of the preceding section. The main objects then fit into the following diagram:
\[
\begin{array}{ccccc}
 L^2(\nu) & \xrightarrow{\ \Phi_S\ } & \cH_S & \xrightarrow{\ \mathcal L\ } & \cH_\Pi\\
 \one & \longmapsto & \mu & \longmapsto & \mathcal L\mu\\[2mm]
 T_{\sig}=\Phi_S^*\Phi_S && A_{\sig}=\Phi_S\Phi_S^* &&
 A_{\pnl}=\mathcal L A_{\sig}\mathcal L^*.
\end{array}
\]
The first row maps design weights to signals and then to PnL. The second follows equal weights along the same route; the third gives the corresponding second-moment operators. These identities do not guarantee that a leading signal eigenvector becomes a leading PnL eigenvector. The payoff map removes some directions and changes the relative size of others, so it can change their spectral ordering.

\begin{proposition}[A sufficient condition for PC transport]\label{prop:payoff-transport}
Under the standing boundedness assumption, suppose $A_{\sig}$ and $D=\mathcal L^*\mathcal L$ have a common orthonormal eigenbasis $(v_j)$ on the closed span of the positive-eigenvalue eigenspaces of $A_{\sig}$, with
\[
  A_{\sig}v_j=\lambda_jv_j,\qquad Dv_j=c_jv_j,
  \qquad \lambda_j>0,\quad c_j\ge0.
\]
Let $v_1$ be a leading signal PC and assume $c_1>0$ and $\lambda_1c_1>\sup_{j\ge2}\lambda_jc_j$ (with empty supremum zero). Then $\mathcal Lv_1$ is a leading eigenvector of $A_{\pnl}$.
\end{proposition}

\begin{proof}
The nonzero images $\mathcal Lv_j/\sqrt{c_j}$ are orthonormal, since $\ip{\mathcal Lv_i}{\mathcal Lv_j}=\ip{v_i}{Dv_j}$. Boundedness and the trace-class expansion of $A_{\sig}$ give
\[
  A_{\pnl}=\sum_j\lambda_j(\mathcal Lv_j)\otimes(\mathcal Lv_j).
\]
Its nonzero eigenvalues are consequently $\lambda_jc_j$. The strict dominance condition proves the claim.
\end{proof}

This is one sufficient set of conditions for signal PC1 to remain leading after the PnL map. Both the common eigenbasis and the ordering of the products $\lambda_jc_j$ need justification in an application.

\section{Counterexamples and structural models}

The additive models in this section allow signals with varying lengths. This makes the distinction between mean and covariance easy to see. Normalizing a nonzero vector as $S_z=X_z/\norm{X_z}$ places it on the unit sphere, but can change its moments. Where we use a normalized version, we state the symmetry needed for the conclusion to survive.

\subsection{The first moment can be orthogonal to PC1}

In a two-dimensional subspace, let $X=(c,Z)$ with $c>0$, $\E Z=0$, and $\E Z^2\gg c^2$. The mean points horizontally, while the leading uncentered second-moment direction is vertical. For a unit-circle version, take $S=X/\norm{X}$ and assume the law of $Z$ is symmetric about zero. The normalized mean remains horizontal, while sufficiently large symmetric two-point values of $Z$ make PC1 vertical. Thus even a population confined to one hemisphere can have its mean perpendicular to PC1.

\subsection{The same PnL PCA, different signal PCA}

In the moving frame let two libraries have identical IC processes $p_z$ but different transverse processes $U_z$ and $\widetilde U_z$. Their PnL kernels are identical by \cref{thm:annihilation}, while their signal kernels differ by
\[
  \int\!\left(\ip{U_{z,t}}{U_{z',t}}-
  \ip{\widetilde U_{z,t}}{\widetilde U_{z',t}}\right)d\tau(t).
\]
Their signal PCs can be arbitrarily different subject to the spherical norm constraints.

\subsection{A fixed mean does not fix the spectrum}

Two libraries can share a mean while their residual covariances, and hence their PCs, differ. For example,
\[
  S_z=m+BF_z+\varepsilon_z,
  \qquad \E_ZF_Z=0,\quad\E_Z\varepsilon_Z=0,
  \quad \E_Z[F_Z\otimes\varepsilon_Z]=0.
\]
Then EWS converges to $m$, whereas PCA orders directions according to
\[
  A_{\sig}=m\otimes m+B\Cov(F)B^*+\E[\varepsilon\otimes\varepsilon].
\]
For $m$ to be PC1, its direction must lead in the combined second moment. If the factor-plus-residual covariance annihilates $m$, this requires $\norm{m}^2$ to exceed its largest transverse eigenvalue. Otherwise, covariance along $m$ also contributes. Under the library law of large numbers, residuals cancel in the average while their covariance remains visible to PCA.

\subsection{A canonical factor can emerge---but from the measure}

Let $r_t$ be a unit direction favored by the research measure---for example, the return-derived direction identified as reversal in the motivating observation---and suppose
\[
  S_{z,t}=\beta_zr_t+\eta_{z,t},
  \qquad \E_Z\eta_{Z,t}=0.
\]
Then $\mu_t=(\E\beta_Z)r_t$. A nonzero mean loading and cancellation of mean residuals make the EWS align with $r_t$. This direction is signal PC1 only under an additional second-moment dominance condition, and it is return PC1 only if it matches the leading eigendirection of the chosen return operator. Its image under the centered payoff map is PnL PC1 only under further payoff-covariance conditions. Neither dimension nor library size supplies these relations.

For a unit-sphere version, set $S_{z,t}=(\beta_zr_t+\eta_{z,t})/\norm{\beta_zr_t+\eta_{z,t}}$ when the denominator is nonzero. If, conditional on $\beta_z$, the residual law is invariant under all orthogonal transformations fixing $r_t$, the normalized population mean lies in $\spanop(r_t)$. Its scalar coefficient need not equal $\E\beta_Z$ and must remain nonzero for a limiting direction to exist.

\section{What low dimension does and does not imply}

At one date the $K\times K$ Gram matrix and the $q\times q$ signal moment both have rank at most $\min(K,q)$. Once $K>q$, the signals must be linearly dependent, although $q$ independent signals already span the available asset space. None of these statements implies clustering, a nonzero mean, or alignment between EWS and a leading signal or return component. Iid uniform points on $S^{q-1}$ have population mean zero and isotropic second moment $I_H/q$. As $K$ grows, their EWS tends to zero and no unique population PC exists. The ratio $K/q$ measures algebraic crowding, not the shape or mean of the sampling measure.

The time dimension also matters. Although each date has rank at most $q$, whole signal processes live in $L^2(\tau;H)$, whose dimension is $qT$ for $T$ dates of positive mass and can be infinite for a population date space. A library can therefore be contemporaneously low-rank yet possess rich temporal structure.

A stronger diversity requirement bounds every pair's time-averaged signal correlation by $\cos\theta$, for a fixed angle $\theta>0$. With equally weighted dates, the normalized history vectors $T^{-1/2}(S_{i,1},\ldots,S_{i,T})$ then form a spherical code with minimum angle $\theta$. At fixed $T$, compactness gives a finite maximum library size, although growing $T$ can permit enormously many separated histories. This excludes exact duplicates, but still does not force a nonzero EWS or alignment with PC1: common components can cancel in the mean while surviving in the second moment. Limits involving library size, history length, and a shrinking separation angle therefore require a separate analysis of both the geometry and the selection rule. These questions are treated in a companion working paper \cite{nunes2026separated}.

\section{Two asymptotic axes: library size and history length}

The population theory above treats temporal inner products as known. In data, they must be estimated from a history of length $T$. Library size $K$ and history length $T$ therefore create separate sources of uncertainty, especially when the library is large relative to the available history.

Let $Y$ be the $T\times K$ matrix of temporally centered PnLs. The empirical PnL covariance is $\widehat C_{\pnl}=T^{-1}Y^\top Y$, so
\[
  \operatorname{rank}(\widehat C_{\pnl})\le \min(K,T-1).
\]
When $K/T\to\gamma\in(0,\infty)$, sample-covariance spectra are not operator-norm consistent under the classical isotropic benchmark. For iid entries with variance $\sigma^2$, the noise spectrum has Marchenko--Pastur upper edge
\begin{equation}\label{eq:mpedge}
  \lambda_+=\sigma^2(1+\sqrt\gamma)^2.
\end{equation}
Large population eigenvalues, often called spikes, can have systematically distorted sample estimates; weak spikes may be indistinguishable from the noise bulk \cite{marchenkopastur1967,johnstone2001,paul2007}. Serial dependence, overlapping forward horizons, changing volatility, and cross-sectional constraints can invalidate the literal iid formula. A useful benchmark must account for those features of the data. Random-matrix filtering has a long history in financial correlation estimation \cite{laloux1999}; shrinkage provides a distinct way to regularize noisy covariance estimates \cite{ledoitwolf2004}. Neither supplies evidence for a population factor without an appropriate sampling comparison.

Time-averaged signal Gram matrices also have temporal estimation error, but not an identical iid aspect ratio. Stacking the $q$ neutral coordinates at each of $T$ dates gives a $qT\times K$ matrix $X$ with $G_{\sig}=T^{-1}X^\top X$ and rank at most $\min(K,qT)$. An idealized independent-coordinate model would involve $K/(qT)$ with appropriate variance scaling; actual coordinates are constrained and grouped by date. Hence both signal and PnL spectra need calibration, with operator-specific rank, dependence, and normalization. Dividing a matrix by $K$, as in the population-operator convention, also divides its noise edge by $K$; all comparisons must use the same scaling.

There are therefore three distinct regimes:
\begin{enumerate}
  \item $K\to\infty$ with population temporal inner products known: the operator limit studied above;
  \item $T\to\infty$ with fixed $K$: classical covariance consistency;
  \item $K,T\to\infty$ jointly with $K/T$ non-negligible: a high-dimensional spectral problem requiring regularization and noise calibration.
\end{enumerate}

\Cref{thm:opconv} controls sampling over designs when each $S_z$ is observed as a population element. It does not by itself control replacement of $\tau$-inner products by a finite time average. In the joint regime, the relevant error contains both a library-sampling term and a temporal-estimation term; the latter need not vanish in operator norm.

Use the iid edge \eqref{eq:mpedge} as a reference only when its assumptions and scaling fit the proposed null. Resampling for uncertainty and constructing a null serve different purposes. A joint block bootstrap, including the stationary bootstrap of \cite{politisromano1994}, retains existing dependence between signals, so it does not represent a world without a common factor.

To test for such a factor, construct surrogate data that remove the dependence being tested. Independent circular shifts or phase randomization across designs are possible choices under suitable stationarity assumptions. For signals, preserve each date's neutral-vector structure. For PnLs, specify which volatility patterns the null retains. Compare the observed eigenvalues and gaps with this benchmark, then assess stability on held-out dates. The resampling scheme must be chosen for the particular hypothesis and data structure.

\section{A numerical illustration and an empirical program}\label{sec:empirical}

The theory suggests diagnostics that distinguish genuine limits from artifacts of multiplicity. In the running 20-asset, 3,000-signal case, $q=19$, the sphere is $S^{18}$, and $K/q\approx158$. The contemporaneous Gram matrix has at most 19 nonzero eigenvalues. Under \cref{cor:axial}, longitudinal variance plus mean pairwise signal correlation must exceed $1/19\approx5.26\%$. That is a substantive hurdle in a 20-asset cross-section: the threshold can bind. With $d=500$ and $q=499$, it falls to $1/499\approx0.2004\%$, a threshold that any cloud retaining even modest longitudinal concentration clears easily. This is why the 20-asset case is informative. The contrast presumes that concentration does not fall with dimension at the same rate as $1/q$; axial symmetry remains a separate hypothesis.

\subsection{Define the sampling unit}

Record how each signal was produced: its generator, feature family, transformation, lookback, training seed, and submission account. Show convergence plots both by raw signal count and with balanced representation across families or clusters. A plot that grows by adding near-duplicates measures replication; it does not establish that new independent variation is being averaged away.

\subsection{Trace four paths as the library grows}

\begin{samepage}
Choose nested library sizes spanning the available range---for the running example, $K=25,50,100,200,500,1000,2000,3000$---and repeatedly randomize the ordering. Estimate:
\begin{enumerate}
  \item $\mu_K(t)$ and its alignment with signal-cloud PC1 and the chosen return-derived principal direction, reporting the observed factor identification separately;
  \item signal Gram eigenvalues/eigenvectors from time-averaged cross-sectional inner products;
  \item $\bar\Pi_K(t)=a_t\bar p_K(t)$ and its Sharpe or other out-of-sample statistic;
  \item PnL covariance and correlation spectra across signals.
\end{enumerate}
\end{samepage}
Use several orderings: uniform over signals, balanced over families, and chronological discovery order. Stability under only one ordering shows that the results depend on how designs are sampled. For each estimated spectrum, report $K$, $T$, $K/T$, matrix rank and scaling, the shrinkage method, and the leading-eigenvalue and eigengap distributions under a specified null. Include the appropriate Marchenko--Pastur reference edge when its assumptions apply. For the stacked signal Gram, also report $q$ and $K/(qT)$; these dimensions help interpret the spectrum but do not account for dependence on their own.

\subsection{Test mean--PC alignment directly}\label{sec:date-alignment}

At a fixed date $t$ with $\mu_K(t)\ne0$, set $m_K=\mu_K(t)/\norm{\mu_K(t)}$, write $M_K=M_{K,t}$, and let $\hat v_{1,K}$ be its unit leading eigenvector. Compute
\begin{align*}
  c_K&=\left|\ip{m_K}{\hat v_{1,K}}\right|,\\
  r_K&=\norm{(I-m_Km_K^\top)M_{K}m_K},\\
  g_K&=m_K^\top M_Km_K-
  \lambda_{\max}((I-m_Km_K^\top)M_K(I-m_Km_K^\top)).
\end{align*}
Here $c_K$ measures alignment, $r_K$ measures how far the mean direction is from being an eigenvector, and $g_K$ compares its energy with the largest transverse energy. Together the residual and gap give a bound for the observed matrix:
\[
  \sin\angle(m_K,\hat v_{1,K})\le \frac{r_K}{g_K}
  \qquad(g_K>0).
\]
To see this, put $P_K=I-m_Km_K^\top$, $B_K=(P_KM_KP_K)|_{m_K^\perp}$, $h_K=P_KM_Km_K$, and $\theta_K=m_K^\top M_Km_K$. Write $\hat v_{1,K}=\alpha m_K+w$ with $w\perp m_K$. The leading eigenvalue $\hat\lambda_1\ge\theta_K$ satisfies
\[
  (\hat\lambda_1 I-B_K)w=\alpha h_K.
\]
If $g_K>0$, then $\alpha\ne0$: otherwise $w=\hat v_{1,K}$ would give $\hat\lambda_1\le\lambda_{\max}(B_K)<\theta_K\le\hat\lambda_1$, a contradiction. The inverse on the transverse block has norm at most $1/g_K$, so $\norm{w}\le|\alpha|r_K/g_K$. Dividing by $|\alpha|$ proves the stronger tangent bound $\tan\angle(m_K,\hat v_{1,K})\le r_K/g_K$. Thus the diagnostics quantitatively bound $c_K$. For asymptotic alignment a sufficient condition is $r_K/g_K\to0$ with $g_K>0$; a small residual alone is insufficient near a degeneracy. Alignment alone does not establish axial symmetry, which additionally predicts a flat transverse spectrum.

To evaluate the empirical counterpart of \eqref{eq:axial-variance}, compute $a_{i,K}=m_K^\top S_{Z_i,t}$ and use $K^{-1}\sum_i(a_{i,K}-\bar a_K)^2+\norm{\mu_K(t)}^2=K^{-1}\sum_i a_{i,K}^2$, with $\bar a_K=\norm{\mu_K(t)}$.\footnote{For $K>1$ unit signals, the mean off-diagonal correlation is $\widehat c_{\rm off}=(K\norm{\mu_K(t)}^2-1)/(K-1)$. Recover the squared empirical centroid as $\norm{\mu_K(t)}^2=[1+(K-1)\widehat c_{\rm off}]/K$ when using this exact empirical identity. Under iid sampling, $\widehat c_{\rm off}$ is an unbiased estimator of the population independent-pair correlation $\norm{\mu_t}^2$; the squared empirical centroid instead has expectation $\norm{\mu_t}^2+(1-\norm{\mu_t}^2)/K$. The difference is diagonal inclusion, not a defect in the off-diagonal population estimator.} Without axial symmetry, exceeding $1/q$ alone does not certify PC1 alignment; the residual and transverse gap above remain the direct tests.

This certificate concerns the signal cloud's $M_K$. Test the motivating return-PC alignment separately by comparing $m_K$ with an estimated $v^{\mathrm{ret}}_1$ from the chosen return operator and window, as defined in \cref{sec:return-pca}. Report its centering convention, $q/T$ aspect ratio, and estimation uncertainty, and use held-out dates for predictive claims. The signal-cloud residual and gap do not certify this return-PC comparison.

\paragraph{A worked geometric example.}
To cross the threshold in \cref{cor:axial} in a controlled way, identify $H$ with $\R^q$, fix a unit pole $m\in H$, and sample
\[
 S=a m+\sqrt{1-a^2}\,U,\qquad 0<a<1,
 \quad U\sim\operatorname{Unif}(S^{q-2}\subset m^\perp).
\]
Here latitude is relative to the mean-signal pole $m$, not to realized returns. Every signal has unit norm and positive latitude. The population mean is $am$, the longitudinal energy is $a^2$, and each of the $q-1$ transverse eigenvalues is $(1-a^2)/(q-1)$. Thus the mean is PC1 exactly when $qa^2>1$. Below the threshold it is orthogonal to the entire leading eigenspace; at equality all $q$ second-moment eigenvalues tie. In the running example, $q=19$: $U$ is uniform on $S^{17}$, each of the 18 transverse eigenvalues is $(1-a^2)/18$, and the threshold is $19a^2>1$. Centered covariance, by contrast, has zero energy along $m$ for every $a$: fixing the latitude removes longitudinal dispersion.

\Cref{fig:threshold} uses $K=3000$ iid transverse directions and sweeps $19a^2$ while retaining the same random directions. This is one reproducible sample path, not a confidence band or a return-data experiment. The transition in population geometry is sharp, while the finite sample's alignment changes over a neighborhood of the threshold. A positive sample eigengap below the threshold does not imply a unique population PC1. Above it, the residual-to-gap ratio bounds the observed angle once $g_K>0$; where $g_K\le0$, the certificate is silent. The supplementary material described below contains the code and numerical diagnostics.

\begin{figure}[!htbp]
\centering
\includegraphics[width=\textwidth]{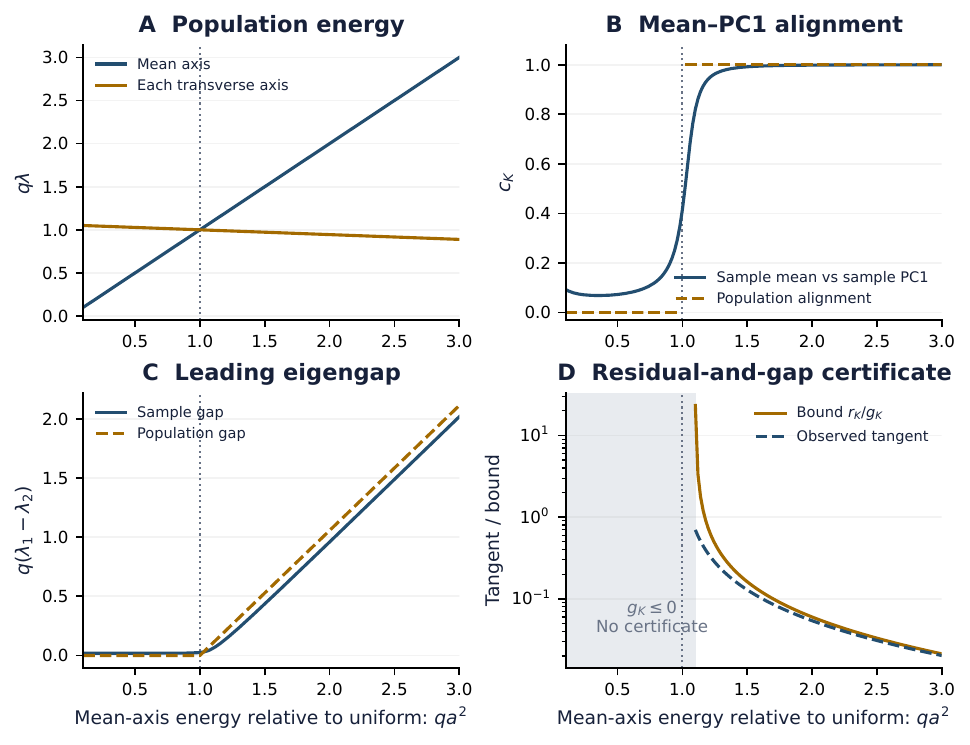}
\caption{Synthetic axial clouds on $S^{18}$, $K=3000$, random seed 20260911. The horizontal variable is $19a^2$, and the dotted line marks the population threshold. (A) Population energies. (B) Sample mean--PC1 cosine and its population counterpart; the latter is undefined at the tie. (C) Sample and population leading eigengaps; the population gap is zero below the threshold because 18 transverse directions tie. (D) The tangent certificate, plotted only where $g_K>0$. All 146 clouds use the same sampled transverse directions; their diagnostics satisfy the bound pointwise.}
\label{fig:threshold}
\end{figure}

\paragraph{Supplementary code.}
\ifarxiv
The simulation script, its numerical outputs (\nolinkurl{threshold.csv}, \nolinkurl{validation.json}), and build instructions accompany this paper as arXiv ancillary files, listed on the abstract page.
\else
The archive \nolinkurl{Large_Signal_Libraries_source.zip} is embedded in this PDF and supplied separately. It contains the manuscript, simulation, figure, diagnostics, and build instructions. Open the paperclip beside this paragraph or use a PDF reader's attachment panel; readers without attachment support can use the separate ZIP.
\fi

\subsection{Test equal centrality}\label{sec:history-alignment}

For the finite signal Gram matrix $G_K$ and PnL covariance matrix $C_K$, inspect the dispersion of row sums, using standard deviations with divisor $K$ and nonzero mean degrees:
\[
  \mathrm{CV}_{\sig}=\frac{\operatorname{sd}(G_K\one)}{|\operatorname{mean}(G_K\one)|},
  \qquad
  \mathrm{CV}_{\pnl}=\frac{\operatorname{sd}(C_K\one)}{|\operatorname{mean}(C_K\one)|}.
\]
\Cref{lem:approxcentrality} turns row-sum variation into an alignment bound once the gap is included. Set $\widetilde G_K=G_K/K$, $u_K=\one/\sqrt K$, and $\theta_K=u_K^\top\widetilde G_Ku_K=\one^\top G_K\one/K^2=\norm{\mu_K}_{\cH_S}^2\ge0$. Then
\[
  \sin\angle(u_K,\hat v_{1,K})
  \le
  \frac{\theta_K\,\mathrm{CV}_{\sig}}
       {\theta_K-\lambda_2(\widetilde G_K)},
\]
provided $\theta_K>\lambda_2(\widetilde G_K)$. Here $\mathrm{CV}_{\sig}$ is computed from the row degrees $G_K\one/K$ (CV is scale invariant), and $\hat v_{1,K}$ now lives in \emph{design-weight space}. If the displayed upper bound is $b_K<1$, \cref{cor:synthesis-transport} gives, for the empirical synthesis map of \cref{app:finite},
\[
  \tan\angle\!\left(\mu_K,
        \frac{\Phi_{S,K}\hat v_{1,K}}{\sqrt{\lambda_1(\widetilde G_K)}}\right)
  \le \sqrt{\frac{\lambda_2(\widetilde G_K)}{\lambda_1(\widetilde G_K)}}
       \frac{b_K}{\sqrt{1-b_K^2}}.
\]
Since $\Phi_{S,K}u_K=\mu_K$, the bound applies to the combined signal history itself, using the same inner product that defines $G_K$. The construction for $C_K$ gives the corresponding bound for centered equal-weight PnL. Its population mean row degree has the direct interpretation
\[
  \theta_{\pnl}=\norm{\Phi_\Pi\one}_{\cH_\Pi}^2
  =\Var_\tau\!\left(\int_{\mathcal Z}\Pi_z\,d\nu(z)\right),
\]
the variance of the equal-weight PnL process. The residual measures dispersion of PnL covariance row degrees, and the gap compares this ensemble variance with the remaining spectrum. If the mean row sum is near zero, use the absolute residual because CV becomes unstable. These bounds describe the observed matrix; conclusions about the population still require temporal calibration. A small CV provides little assurance when the gap is also small.

The diagnostics answer different questions. \Cref{sec:date-alignment} compares the EWS with the signal cloud's leading asset direction at one date. \Cref{sec:history-alignment} compares fixed design weights and the histories they produce. Close alignment over a whole history can conceal poor alignment on particular dates. Conversely, daily alignment need not imply one leading weight pattern for the entire history. If the diagnostics disagree, examine the angles and gaps by regime and each regime's contribution to the history norm.

\subsection{Exploit the moving-frame decomposition}

Estimate the signal kernel as $K_\parallel+K_\perp$ using realized IC products and transverse inner products. Compare its spectrum with that of $K_\parallel$ to see what is lost when the transverse components are removed. Next compare $K_\parallel$ with the uncentered, dispersion-weighted IC Gram matrix, and then with centered PnL covariance. These steps separate the effects of projection, dispersion weighting, and temporal centering.

\subsection{Respect time dependence and selection}

Use rolling or blocked out-of-sample windows, block or stationary bootstrap resampling, and family-level resampling. Report eigengaps with uncertainty. Factor alignments, PCs, and Sharpe ratios chosen after inspecting thousands of signals are selection-affected statistics; the asymptotic theory does not remove data-snooping bias or replace multiple-testing corrections \cite{white2000,harvey2016}.

\section{Interpretive consequences}

\paragraph{Library growth alone does not imply EWS--PC1 alignment.}
Without assumptions on how designs are added, the library mean need not converge at all. When it does, alignment with signal-cloud PC1 is governed by \cref{thm:meanpc,cor:axial}. Alignment with a return component needs a further relation between the research distribution and returns. Reversal identifies the dominant direction in the motivating observation; it is the relation between the research distribution and returns, rather than the signal count, that needs explaining.

\paragraph{EWS and PC1 answer different questions.}
EWS asks which signed direction survives averaging. PC1 asks which direction carries the most squared projection. A population can have zero mean and a strong principal axis, or a nonzero mean perpendicular to that axis, as the examples above show.

\paragraph{Signal PCA and PnL PCA are not competing estimators.}
Signal PCA measures similarity between signal histories, including components orthogonal to returns. PnL PCA measures covariance between payoff histories. Shared transverse structure may dominate the former and disappear from the latter. Changes in return magnitude and temporal centering can also reshape the PnL spectrum.

\paragraph{Low-reversal sub-libraries need not be higher PCs.}
Selecting signals with low reversal correlation changes the design distribution being averaged. The subset may combine several PCs, and both its mean and its row-sum pattern can differ from those of the full library. A correlation constraint does not identify an eigenvector.

\paragraph{Normalization changes the reported ensemble.}
The raw EWS converges linearly to $\mu$. A date-by-date renormalized EWS converges to $\mu_t/\norm{\mu_t}$ only where $\mu_t\ne0$. Near zeros, small estimation errors can cause large directional changes. Unit-gross, volatility, and turnover normalizations introduce further nonlinearities not covered by the simplest mean theorem.

\section{Conclusion}

Under the relevant assumptions, equal weights estimate a mean and signal PCA estimates a second moment or covariance. Equal-weight PnL follows the mean IC process scaled by return magnitude. PnL PCA describes covariance after projection onto returns, dispersion weighting, and temporal centering. Agreement between these objects requires the structural conditions developed above, such as mean--eigenvector compatibility, a positive eigengap, equal kernel row sums, or the stated symmetry and factor assumptions.

On $S^{q-1}$, axial symmetry makes the geometric dividing line explicit: the nonzero mean direction is signal-cloud PC1 exactly when its longitudinal energy exceeds the isotropic share $1/q$. Large $K/q$ by itself supplies no such concentration. Explaining the reported alignment with a leading return component further requires a relation between the signal-design measure and the return operator. Date-by-date signal geometry, whole-history ensemble alignment, and return-factor alignment are distinct empirical questions; finite-history calibration is required for population inference.

\begin{samepage}
\phantomsection\addcontentsline{toc}{section}{Acknowledgments and AI-assistance disclosure}
\section*{Acknowledgments and AI-assistance disclosure}

The author developed and revised this paper through iterative dialogue with ChatGPT, an AI system developed by OpenAI. The system assisted with mathematical organization, counterexamples, and manuscript production. Claude, an AI system developed by Anthropic, was used for editorial polishing, bibliography verification, and proofreading of the final drafts. Neither system is an author or independent reviewer. Marc Nunes directed the inquiry, selected the claims, and accepts full responsibility for the paper's content, accuracy, and conclusions.
\end{samepage}

\appendix
\section{Finite-library matrix identities}\label{app:finite}

For observed designs $Z_1,\ldots,Z_K$, define the synthesis map
\[
  \Phi_{S,K}v:=\frac1{\sqrt K}\sum_{i=1}^K v_iS_{Z_i}.
\]
Then
\[
  \Phi_{S,K}^*\Phi_{S,K}=\frac1K G_{\sig,K},
  \qquad
  \Phi_{S,K}\Phi_{S,K}^*=A_{\sig,K}.
\]
Thus the scaled Gram matrix and feature-space empirical operator share nonzero eigenvalues. The equal-weight signal is
\[
  \mu_K=\frac1K\sum_iS_{Z_i}=\frac1{\sqrt K}\Phi_{S,K}\one.
\]
If $G_{\sig,K}\one=c_K\one$, then
\[
  A_{\sig,K}\mu_K=(c_K/K)\mu_K.
\]
The finite PnL identities are identical with $S_{Z_i}$ replaced by $\Psi_{Z_i}$.

\section{A normalization warning}

Let $D$ be the diagonal matrix of PnL standard deviations. PnL correlation is $R=D^{-1}CD^{-1}$. Even if $C\one=c\one$, it need not follow that $R\one$ is constant unless the marginal volatilities are compatible with the covariance structure. PCA of covariance, PCA of correlation, and risk-normalized combination should therefore be reported separately.

\phantomsection


\addcontentsline{toc}{section}{References}
\begin{thebibliography}{99}
\small\raggedright
\setlength{\itemsep}{2pt}
\setlength{\parsep}{0pt}
\bibitem{batesgranger1969}
Bates, J. M. and Granger, C. W. J. (1969). The combination of forecasts. \emph{Operational Research Quarterly} 20(4), 451--468. \href{https://doi.org/10.1057/jors.1969.103}{\nolinkurl{doi:10.1057/jors.1969.103}}.

\bibitem{bosq2000}
Bosq, D. (2000). \emph{Linear Processes in Function Spaces: Theory and Applications}. Springer. \href{https://doi.org/10.1007/978-1-4612-1154-9}{\nolinkurl{doi:10.1007/978-1-4612-1154-9}}.

\bibitem{dauxois1982}
Dauxois, J., Pousse, A., and Romain, Y. (1982). Asymptotic theory for the principal component analysis of a vector random function: some applications to statistical inference. \emph{Journal of Multivariate Analysis} 12(1), 136--154. \href{https://doi.org/10.1016/0047-259X(82)90088-4}{\nolinkurl{doi:10.1016/0047-259X(82)90088-4}}.

\bibitem{daviskahan1970}
Davis, C. and Kahan, W. M. (1970). The rotation of eigenvectors by a perturbation. III. \emph{SIAM Journal on Numerical Analysis} 7(1), 1--46. \href{https://doi.org/10.1137/0707001}{\nolinkurl{doi:10.1137/0707001}}.

\bibitem{demiguel2009}
DeMiguel, V., Garlappi, L., and Uppal, R. (2009). Optimal versus naive diversification: How inefficient is the $1/N$ portfolio strategy? \emph{Review of Financial Studies} 22(5), 1915--1953. \href{https://doi.org/10.1093/rfs/hhm075}{\nolinkurl{doi:10.1093/rfs/hhm075}}.

\bibitem{dingmartin2017}
Ding, Z. and Martin, R. D. (2017). The Fundamental Law of Active Management: Redux. \emph{Journal of Empirical Finance} 43, 91--114. \href{https://doi.org/10.1016/j.jempfin.2017.05.005}{\nolinkurl{doi:10.1016/j.jempfin.2017.05.005}}.

\bibitem{grinold1994}
Grinold, R. C. (1994). Alpha is volatility times IC times score. \emph{Journal of Portfolio Management} 20(4), 9--16. \href{https://doi.org/10.3905/jpm.1994.409482}{\nolinkurl{doi:10.3905/jpm.1994.409482}}.

\bibitem{harvey2016}
Harvey, C. R., Liu, Y., and Zhu, H. (2016). \ldots\ and the cross-section of expected returns. \emph{Review of Financial Studies} 29(1), 5--68. \href{https://doi.org/10.1093/rfs/hhv059}{\nolinkurl{doi:10.1093/rfs/hhv059}}.

\bibitem{jegadeesh1990}
Jegadeesh, N. (1990). Evidence of predictable behavior of security returns. \emph{Journal of Finance} 45(3), 881--898. \href{https://doi.org/10.1111/j.1540-6261.1990.tb05110.x}{\nolinkurl{doi:10.1111/j.1540-6261.1990.tb05110.x}}.

\bibitem{johnstone2001}
Johnstone, I. M. (2001). On the distribution of the largest eigenvalue in principal components analysis. \emph{Annals of Statistics} 29(2), 295--327. \href{https://doi.org/10.1214/aos/1009210544}{\nolinkurl{doi:10.1214/aos/1009210544}}.

\bibitem{kallenberg2005}
Kallenberg, O. (2005). \emph{Probabilistic Symmetries and Invariance Principles}. Springer. \href{https://doi.org/10.1007/0-387-28861-9}{\nolinkurl{doi:10.1007/0-387-28861-9}}.

\bibitem{koltchinskii2000}
Koltchinskii, V. and Gin\'e, E. (2000). Random matrix approximation of spectra of integral operators. \emph{Bernoulli} 6(1), 113--167. \href{https://doi.org/10.2307/3318636}{\nolinkurl{doi:10.2307/3318636}}.

\bibitem{laloux1999}
Laloux, L., Cizeau, P., Bouchaud, J.-P., and Potters, M. (1999). Noise dressing of financial correlation matrices. \emph{Physical Review Letters} 83(7), 1467--1470. \href{https://doi.org/10.1103/PhysRevLett.83.1467}{\nolinkurl{doi:10.1103/PhysRevLett.83.1467}}.

\bibitem{ledoitwolf2004}
Ledoit, O. and Wolf, M. (2004). A well-conditioned estimator for large-dimensional covariance matrices. \emph{Journal of Multivariate Analysis} 88(2), 365--411. \href{https://doi.org/10.1016/S0047-259X(03)00096-4}{\nolinkurl{doi:10.1016/S0047-259X(03)00096-4}}.

\bibitem{lehmann1990}
Lehmann, B. N. (1990). Fads, martingales, and market efficiency. \emph{Quarterly Journal of Economics} 105(1), 1--28. \href{https://doi.org/10.2307/2937816}{\nolinkurl{doi:10.2307/2937816}}.

\bibitem{marchenkopastur1967}
Marchenko, V. A. and Pastur, L. A. (1967). Distribution of eigenvalues for some sets of random matrices. \emph{Mathematics of the USSR-Sbornik} 1(4), 457--483. \href{https://doi.org/10.1070/SM1967v001n04ABEH001994}{\nolinkurl{doi:10.1070/SM1967v001n04ABEH001994}}.

\bibitem{mardia2000}
Mardia, K. V. and Jupp, P. E. (2000). \emph{Directional Statistics}. Wiley. \href{https://doi.org/10.1002/9780470316979}{\nolinkurl{doi:10.1002/9780470316979}}.

\bibitem{nunes2026}
Nunes, M. (2026). \emph{Signal Correlation, IC, and PnL Dependence}. \href{https://arxiv.org/abs/2609.09588}{arXiv:2609.09588}.

\bibitem{nunes2026separated}
Nunes, M. (2026). \emph{Separated Signal Libraries: Packing, Saturation, and Joint Spectral Limits}. Working paper, in preparation.

\bibitem{paul2007}
Paul, D. (2007). \href{https://www3.stat.sinica.edu.tw/statistica/J17N4/J17N418/J17N418.html}{Asymptotics of sample eigenstructure for a large dimensional spiked covariance model}. \emph{Statistica Sinica} 17(4), 1617--1642.

\bibitem{politisromano1994}
Politis, D. N. and Romano, J. P. (1994). The stationary bootstrap. \emph{Journal of the American Statistical Association} 89(428), 1303--1313. \href{https://doi.org/10.1080/01621459.1994.10476870}{\nolinkurl{doi:10.1080/01621459.1994.10476870}}.

\bibitem{qianhua2004}
Qian, E. and Hua, R. (2004). \href{https://www.panagora.com/assets/JOIM-Active-Risk-and-Information-Ratio.pdf}{Active risk and information ratio}. \emph{Journal of Investment Management} 2(3), 20--34.

\bibitem{rosasco2010}
Rosasco, L., Belkin, M., and De Vito, E. (2010). \href{https://www.jmlr.org/papers/v11/rosasco10a.html}{On learning with integral operators}. \emph{Journal of Machine Learning Research} 11, 905--934.

\bibitem{saad2011}
Saad, Y. (2011). \emph{Numerical Methods for Large Eigenvalue Problems}, 2nd ed. SIAM. \href{https://doi.org/10.1137/1.9781611970739}{\nolinkurl{doi:10.1137/1.9781611970739}}.

\bibitem{schaefer1974}
Schaefer, H. H. (1974). \emph{Banach Lattices and Positive Operators}. Springer. \href{https://doi.org/10.1007/978-3-642-65970-6}{\nolinkurl{doi:10.1007/978-3-642-65970-6}}.

\bibitem{scholkopf1998}
Sch\"olkopf, B., Smola, A., and M\"uller, K.-R. (1998). Nonlinear component analysis as a kernel eigenvalue problem. \emph{Neural Computation} 10(5), 1299--1319. \href{https://doi.org/10.1162/089976698300017467}{\nolinkurl{doi:10.1162/089976698300017467}}.

\bibitem{timmermann2006}
Timmermann, A. (2006). Forecast combinations. In G. Elliott, C. W. J. Granger, and A. Timmermann (eds.), \emph{Handbook of Economic Forecasting}, Vol. 1, 135--196. Elsevier. \href{https://doi.org/10.1016/S1574-0706(05)01004-9}{\nolinkurl{doi:10.1016/S1574-0706(05)01004-9}}.

\bibitem{white2000}
White, H. (2000). A reality check for data snooping. \emph{Econometrica} 68(5), 1097--1126. \href{https://doi.org/10.1111/1468-0262.00152}{\nolinkurl{doi:10.1111/1468-0262.00152}}.

\bibitem{yuwangsamworth2015}
Yu, Y., Wang, T., and Samworth, R. J. (2015). A useful variant of the Davis--Kahan theorem for statisticians. \emph{Biometrika} 102(2), 315--323. \href{https://doi.org/10.1093/biomet/asv008}{\nolinkurl{doi:10.1093/biomet/asv008}}.

\end{thebibliography}
\end{document}